%% file: gomory-hu.tex
\documentclass[11pt]{article}
\usepackage{amsmath,amssymb,amsfonts,latexsym,graphicx,amsthm}
\usepackage{fullpage,color}
\usepackage{url,hyperref}
\usepackage{comment}
\usepackage[linesnumbered,boxed,ruled,vlined]{algorithm2e}
\usepackage{framed}
\usepackage[shortlabels]{enumitem}

\usepackage[margin=1in]{geometry}

\newtheorem{theorem}{Theorem}[section]
\newtheorem{lemma}{Lemma}[section]
\newtheorem{definition}{Definition}[section]
\newtheorem{corollary}{Corollary}[section]
\newtheorem{claim}{Claim}[lemma]

\newtheorem{hypothesis}{Hypothesis}[section]

        {\medskip}

\newcommand{\etal}{{\em et al.}\ }

\newcommand{\out}{\mathsf{out}}
\newcommand{\vol}{\mathsf{vol}}
\newcommand{\gh}{\mathcal{T}}
\newcommand{\cnt}{\mathsf{cnt}}

\newcommand{\const}{\mathsf{c}}
\newcommand{\mf}{\mathsf{MF}}
\newcommand{\old}{\mathrm{old}}

\begin{document}

\title{Faster Cut-Equivalent Trees in Simple Graphs}
\author{Tianyi Zhang \thanks{Tel Aviv University, \href{}{tianyiz21@tauex.tau.ac.il}}}
\date{}

\maketitle

\input{intro}
\input{overview}

\input{prelim}
\input{conditional}
\input{expander}
\input{correctness}
\input{time}
\input{unconditional}
\input{ack}

\vspace{5mm}
\bibliographystyle{alpha}
\bibliography{ref}

\end{document}

%% file: intro.tex
\begin{abstract}
	Let $G = (V, E)$ be an undirected connected simple graph on $n$ vertices. A cut-equivalent tree of $G$ is an edge-weighted tree on the same vertex set $V$, such that for any pair of vertices $s, t\in V$, the minimum $(s, t)$-cut in the tree is also a minimum $(s, t)$-cut in $G$, and these two cuts have the same cut value. In a recent paper [Abboud, Krauthgamer and Trabelsi, STOC 2021], the authors propose the first subcubic time algorithm for constructing a cut-equivalent tree. More specifically, their algorithm has \footnote{$\widetilde{O}$ hides poly-logarithmic factors.}$\widetilde{O}(n^{2.5})$ running time. Later on, this running time was significantly improved to $n^{2+o(1)}$ by two independent works [Abboud, Krauthgamer and Trabelsi, FOCS 2021] and [Li, Panigrahi, Saranurak, FOCS 2021], and then to $(m+n^{1.9})^{1+o(1)}$ by [Abboud, Krauthgamer and Trabelsi, SODA 2022].
	
	In this paper, we improve the running time to $\widetilde{O}(n^2)$ graphs if near-linear time max-flow algorithms exist, or $\widetilde{O}(n^{17/8})$ using the currently fastest max-flow algorithm. Although our algorithm is slower than previous works, the runtime bound becomes better by a sub-polynomial factor in dense simple graphs when assuming near-linear time max-flow algorithms.
\end{abstract}

\section{Introduction}
It is well known from Gomory and Hu \cite{gomory1961multi} that any undirected graph can be compressed into a single tree while all pairwise minimum cuts are preserved exactly. More specifically, given any undirected graph $G = (V, E)$ on $n$ vertices and $m$ edges, there exists an edge weighted tree $T$ on the same set of vertices $V$, such that: for any pair of vertices $s, t\in V$, the minimum $(s, t)$-cut in $T$ is also a minimum $(s, t)$-cut in $G$, and their cut values are equal. Such trees are called Gomory-Hu trees or cut-equivalent trees. In the original paper \cite{gomory1961multi}, Gomory and Hu showed an algorithm that reduces the task of building a cut-equivalent tree to $n-1$ max-flow instances. Gusfield \cite{gusfield1990very} modified the original algorithm Gomory and Hu so that no graph contractions are needed when applying max-flow subroutines. So far, in weighted graphs, faster algorithms for building cut-equivalent trees were only byproducts of faster max-flow algorithms. In the recent decade, there has been a sequence of improvements on max-flows using the interior point method \cite{lee2014path,madry2016computing,liu2020faster,kathuria2020unit,brand2021minimum}, and the current best running time is $\widetilde{O}(m + n^{1.5})$ by \cite{brand2021minimum}, so computing a cut-equivalent tree takes time $\widetilde{O}(mn + n^{2.5})$.

When $G$ is a simple graph, several improvements have been made over the years. Bhalgat \etal \cite{hariharan2007mn} designed an $\widetilde{O}(mn)$ time algorithm for cut-equivalent trees using a tree packing approach based on \cite{gabow1995matroid,edmonds2003submodular}. Recent advances include an upper bound of $O(m^{3/2}n^{1/6})$ by Abboud, Krauthgamer and Trabelsi \cite{abboud2020new}, and in a subsequent work \cite{abboud2020subcubic} by the same set of authors, they proposed the first subcubic time algorithm that constructs cut-equivalent trees in simple graphs, and their running time is $n^{2.5+o(1)}$. Recently, by two independent works \cite{abboud2021apmf,li2021nearly}, this running time was improved to $n^{2+o(1)}$ which is almost-optimal in dense graphs, and further to a  subquadratic time $(m+n^{1.9})^{1+o(1)}$ by \cite{abboud2022friendly}.

All of these upper bounds rely on the current fastest max-flow algorithm with runtime $\widetilde{O}(m+n^{1.5})$. However, even if we assume the existence of a $\widetilde{O}(m)$-time max-flow algorithm, the above algorithms still have $n^{2+o(1)}$ running time in dense graphs which contains an extra sub-polynomial factor.

\subsection{Our results}
Let $\mf(m_0, n_0)$ be the running time complexity of max-flow computation in unweighted multi-graphs with $m_0$ edges and $n_0$ vertices, and let $\mf(m_0) = \mf(m_0, m_0)$ for convenience.

The main result of this paper is a near-quadratic time algorithm assuming existence of quasi-linear time max-flow algorithms. For a detailed comparison with recent published works, please refer to the table below where conditional runtime refers to the assumption of near-linear time max-flow algorithms.

\begin{hypothesis}\label{hypo}
	$\mf(m_0, n_0) = \widetilde{O}(m_0 + n_0)$.
\end{hypothesis}

\begin{theorem}\label{cond}
	Let $G = (V, E)$ be a simple on $n$ vertices. Under Hypothesis~\ref{hypo}, there is a randomized algorithm that constructs a cut-equivalent tree of $G$ in $\widetilde{O}(n^2)$ time with high probability. Using the current fastest max-flow algorithm~\cite{brand2021minimum}, the running time becomes $\widetilde{O}(n^{17/8})$.
\end{theorem}

\begin{center}
	\begin{tabular}{|c|c|c|}\hline
		reference	&	conditional runtime	&	unconditional runtime\\\hline
		\cite{abboud2020subcubic}	&	$\widetilde{O}(n^{2.5})$	&	$n^{2.5+o(1)}$\\\hline
		\cite{abboud2021apmf}	&	$n^{2+o(1)}$	&	$n^{2+o(1)}$\\\hline
		\cite{li2021nearly}		&	$n^{2+o(1)}$	&	$n^{2+o(1)}$\\\hline
		\cite{abboud2022friendly}	&	$(m+n^{1.75})^{1+o(1)}$	&	$(m+n^{1.9})^{1+o(1)}$\\\hline
		new	&	$\widetilde{O}(n^2)$	&	$\widetilde{O}(n^{17/8})$\\\hline
	\end{tabular}
\end{center}

\noindent\textbf{Comparison with subsequent works.} In a very recent but unpublished online preprint~\cite{abboud2021gomory} (see also a note by \cite{zhang2021gomory}), significant progress has been made where an unconditional $\widetilde{O}(n^2)$ runtime has been achieved for cut-equivalent trees in general weighted graphs, which completely subsumes our result.

%% file: overview.tex
\subsection{Technical overview}
Our algorithm is largely based on the framework of \cite{abboud2020subcubic}. In this subsection, we will discuss the running time bottlenecks of \cite{abboud2020subcubic} and how to bypass them. For simplification, consider the following task. Let $\gh$ be a partition tree which is an intermediate tree of the Gomory-Hu algorithm. Take an arbitrary node $N\subseteq V$ of $\gh$ which represents a vertex subset of $V$. Let $G_\gh[N] = (V_\gh[N], E_\gh[N])$ be the auxiliary graph obtained by contracting each component of $\gh\setminus \{N\}$ into a single vertex in the original graph $G$.

Fix a pivot vertex $p\in N$, we want to find a sequence of vertices $v_1, v_2, \cdots, v_l\in N$, and compute a sequence of latest minimum cuts $(L_i, V_\gh[N]\setminus L_i), 1\leq i\leq l$ in $G_\gh[N]$ for $(v_i, p), 1\leq i\leq l$, where $v_i\in K_i, p\notin K_i$, such that:
\begin{enumerate}[(1)]
	\item $l\geq \Omega(|N|)$.
	\item For each $1\leq i\leq l$, $|L_i| \leq |N|/2$.
\end{enumerate}
If we cut all sides $L_i\cap N$ off of $N$ and form tree nodes, then by the above two properties all tree nodes are vertex subsets of $N$ of size at most $|N|/2$. So, if we can recursively repeat this procedure on smaller subsets, then it would produce a cut-equivalent tree in logarithmically many rounds.

For this task, the basic idea of \cite{abboud2020subcubic} is to apply expander decompositions. Suppose the original graph $G$ is decomposed into disjoint clusters $V = C_1\cup C_2\cup\cdots $ such that each $G[C_i]$ is a $\phi$-expander, and the total number of inter-cluster edges is bounded by $\widetilde{O}(\phi n^2)$. For simplicity let us assume $G$ is a roughly regular graph and each vertex $v\in V$ has degree $\deg_G(v)\in [n/2, n)$. For each $v$, suppose $(L_v, V\setminus L_v)$ is the latest min-cut for $(v, p)$, and let $C_v$ be the $\phi$-expander of the expander decomposition that contains $v$. Vertices of $N$ are divided into three types.
\begin{enumerate}
	\item Vertices in clusters whose size is less than $n/4$, namely $|C_v| < n/4$.
	\item Vertices in clusters whose size is at least $n/4$, namely $|C_v|\geq n/4$, plus that $|C_v\cap L_v| \leq 10 / \phi$. Note that there are only a constant number of such clusters.
	\item $|C_v| \geq n/4$, plus that $|C_v\setminus L_v| \leq 10 / \phi$.
\end{enumerate}

\paragraph*{The first bottleneck} To compute $L_v$ for type-1 vertices, we simply go over all such $v$'s, and compute the max-flow from $v$ to $p$ in $G_\gh[N]$. Since each type-1 vertex must contribute $n/2 - n/4 = n/4$ inter-cluster edges as the input graph $G$ is simple, the total number of type-1 vertices does not exceed $\widetilde{O}(\phi n)$, summing over all tree nodes $N$ of $\gh$.

For type-2 vertices, using the isolating cut lemma devised in \cite{abboud2020subcubic,li2020deterministic}, we can compute all the sides $L_v$ in $\widetilde{O}(\mf(\vol_G(N)) / \phi)$ time, which sum to $\widetilde{O}(\mf(n^2) / \phi)$ over all nodes $N$ of $\gh$. So, under Hypothesis~\ref{hypo}, the total time cost of type-1 and type-2 vertices is $\widetilde{O}(\phi n^3 + n^2 / \phi)$, which is always larger than $n^{2.5}$. So in their algorithm~\cite{abboud2020subcubic}, parameter $\phi$ is equal to $1/\sqrt{n}$.

Our observation is that applying max-flow for each type-1 vertex is too costly. To overcome this bottleneck, we simply avoid computing cuts $(L_v, V_\gh[N]\setminus L_v)$ for both type-1 and type-2 vertices. If the total number of type-2\&3 vertices is larger than the total number of type-1 vertices, then we can skip all type-1 vertices. However, if the number of type-1 vertices dominates in $N$, then the number of type-2\&3 vertices is at most $\widetilde{O}(\phi n)$ over all such kind of $N$. In this case, the total degree $\vol_G(N)$ is at most $\widetilde{O}(\phi n^2)$, and therefore, when summing over all nodes $N$ of $\gh$, computing all type-1 vertices takes time at most $\widetilde{O}(\phi^2 n^3)$, instead of $\widetilde{O}(\phi n^3)$, and so the new balance would be $\widetilde{O}(\phi^2n^3 + n^2/\phi)$. Therefore, if we choose $\phi = n^{-1/3}$, it becomes $n^{7/3}$ which is already better than $n^{2.5}$. In the final algorithm, we will classify expander sizes using $\log n$ many different thresholds, instead of just one threshold (which is n/4 here), and so in the end we can set $\phi = 1 / \log^{O(1)}n$. In general cases where graph $G$ has various vertex degrees, we need to apply boundary-linked expander decomposition from a recent work \cite{goranci2021expander}; especially we need to make use of property (3) in Definition 4.2 of \cite{goranci2021expander}.

\paragraph*{The second bottleneck} In the work \cite{abboud2020subcubic}, in order to compute latest min-cuts $(L_v, V_\gh[N]\setminus L_v)$ for type-3 vertices, they consider the laminar family formed by all sides $L_v$. If the laminar family has tree depth at most $k$, then their algorithm can compute cuts $(L_v, V_\gh[N]\setminus L_v)$ by applying $k + 10/\phi$ max-flows in $G_\gh[N]$. To ensure that the depth is bounded by $k$, they need a first randomly refine node $N$ into $|N| / k$ sub-nodes which takes $|N|/k$ Gomory-Hu steps. Hence, in total, it requires at least $k + |N|/k > \sqrt{|N|}$ max-flow invocations, which leads to a $n^{2.5}$ running time under Hypothesis~\ref{hypo}.

To bypass this barrier, the observation is that the depth of the laminar family in each $\phi$-expander is already small, so actually we do not need the help from the refinement step. More precisely, instead of looking at the entire laminar family formed by sets $\{L_v\}_{v\in N}$, we only look at the laminar family formed by sets $\{C_v\cap L_v \}_{v\in N}$ for each cluster $C$. It can be proved that the depth of this smaller laminar family is always bounded by $O(1/\phi)$. In the end, to compute latest min-cuts $(L_v, V_\gh[N]\setminus L_v)$ for all type-2\&3 vertices, we will only use $O(1/\phi)$ max-flow instances in total.

%% file: prelim.tex
\section{Preliminaries}
Let $G = (V, E)$ be an arbitrary simple graph on $n$ vertices and $m$ edges with unit-capacities. For any $v\in V$, let $\deg_G(v)$ be the number of its neighbors in $V$. For any subset $S\subseteq V$, define $\vol_G(S) = \sum_{v\in S}\deg_G(v)$, and let $\out_G(S)$ count the number of edges in $E\cap (S\times (V\setminus S))$, and define $G[S]$ to be the induced subgraph of $S$ on $G$. 

Introduced in \cite{gabow1991applications}, the latest minimum $(s, t)$-cut is a minimum $(s, t)$-cut such that the side containing $s$ has minimum size as well. It is proved that latest minimum cuts are unique, and can be computed by any max-flow algorithm for $(s, t)$. 

Here are some basic facts about min-cuts.
\begin{lemma}[Lemma 2.8 in \cite{abboud2020subcubic}]\label{union}
	For any vertices $a, b, p\in V$, assume $(A, V\setminus A)$ and $(B, V\setminus B)$ are min-cuts for $(a, p), (b, p)$ respectively. If $b\in A$, then $(A\cup B, V\setminus (A\cup B))$ is a min-cut for $(a, p)$ as well.
\end{lemma}
\begin{lemma}[Lemma 2.9 in \cite{abboud2020subcubic}]\label{minus}
	For any vertices $a, b, p\in V$, assume $(A, V\setminus A)$ and $(B, V\setminus B)$ are min-cuts for $(a, p), (b, p)$ respectively. If $a\notin B$ and $b\notin A$, then $(A\setminus B, V\setminus (A\setminus B))$ is a min-cut for $(a, p)$. 
\end{lemma}

\subsection{Cut-equivalent trees}
A cut-equivalent tree is a tree $\gh$ on $V$ with weighted edges, such that for any pair $s, t\in V$, there is a minimum cut $(S, V\setminus S)$ in $G$ such that it is also a minimum cut in $\gh$ with the same cut value. Now let us turn to define some terminologies for cut-equivalent trees.

\paragraph*{Partition trees} A partition tree $\gh$ of graph $G$ is a tree whose nodes $U_1, U_2, \cdots, U_l$ represent disjoint subsets of $V$ such that $V = U_1\cup U_2\cup \cdots \cup U_l$. For each node $U$ of $\gh$, the auxiliary graph $G_\gh[U] = (V_\gh[U], E_\gh[U])$ of $U$ is built by contracting each component of $\gh\setminus \{U \}$ into a single vertex in the original graph $G$.

\paragraph*{Gomory-Hu algorithm} Gomory-Hu algorithm provides a flexible framework for constructing a cut-equivalent tree. The algorithm begins with a partition tree $\gh$ which is the single node that subsumes the entire vertex set $V$, and creates more nodes iteratively by refining its nodes. In each iteration, the algorithm picks an arbitrary node that represent a non-singleton subset $U\subseteq V$, and selects two arbitrary vertex $s, t\in U$. Then, compute the minimum cut $(S, V_\gh[U]\setminus S)$ between $s, t$ in the auxiliary graph $G_\gh[U]$. Finally, split node $U$ into two nodes that correspond to subsets $S\cap U$ and $(V_\gh[U]\setminus S)\cap U$ respectively, connected by an edge with weight equal to the value of the cut $(S, V_\gh[U]\setminus S)$ in $G_\gh[U]$. For each node $W$ that was $U$'s neighbor on $\gh$, reconnect $W$ to node $S\cap U$ if $S$ contains the contracted node that subsumes $W$; otherwise reconnect $W$ to node $(V_\gh[U]\setminus S)\cap U$.

A tree is called \textbf{GH-equivalent}, if it is a partition tree that can be constructed during Gomory-Hu algorithm by certain choice of nodes $U$ to split and pairs of vertices $s, t\in U$.

\paragraph*{Refinement with respect to subsets}  Consider a partition tree $\gh$ which is GH-equivalent. Let $U$ be one of $\gh$'s node and let $R\subseteq U$ be a subset. A refinement of $\gh$ with respect to $R$ is to repeatedly execute a sequence of Gomory-Hu iterations by always picking two different vertices from $s, t\in R$ that are currently in the same node of $\gh$ and refine $\gh$ using a minimum $(s, t)$-cut. So after the refinement of $\gh$ with respect to $R$, $\gh$ is still GH-equivalent.

\begin{lemma}[\cite{granot1986multi}]\label{refine}
	For any node $U$ of $\gh$ and any subset $R\subseteq U$, a refinement of $\gh$ with respect to $R$ can be computed in time $\widetilde{O}(|R|\cdot \mf(\vol_G(U)))$.
\end{lemma}

\begin{lemma}[see definition of partial trees in \cite{abboud2020subcubic}]\label{partial-tree}
	After the refinement on a GH-equivalent tree $\gh$ with respect to $R$, for any $a, b\in R$, let $N_a\ni a$ and $N_b\ni b$ be nodes of $\gh$. Then, the min-cut of $(N_a, N_b)$ in $\gh$ is a min-cut in $G$ for $(a, b)$ as well.
\end{lemma}

\paragraph*{$k$-partial trees} A $k$-partial tree $\gh$ is a GH-equivalent partition tree such that all vertices $u\in V$ such that $\deg_G(u)\leq k$ are singletons of $\gh$. The following lemma states that $k$-partial trees always exist and can be computed efficiently for small $k$'s.

\begin{lemma}[\cite{hariharan2007mn}]\label{k-partial}
	There is an algorithm that, given an undirected graph with unit edge capacities and parameter $k$ on $n$ vertices, computes a $k$-partial tree in time $\min\{\widetilde{O}(nk^2), \widetilde{O}(mk)\}$.
\end{lemma}

\subsection{Expander decomposition}
For any pair of disjoint sets $S, T\subseteq V$, let $E_G(S, T)$ be the set of edges between $S, T$ in $G$. The conductance of a cut $(S, V\setminus S)$ is $\Phi_G(S) = |E_G(S, V\setminus S)| / \min\{\vol_G(S), \vol_G(V\setminus S) \}$, and the conductance of a graph $G$ is defined as $\Phi_G = \min_{S} \Phi_G(S)$. A graph $G$ is a $\phi$-expander if $\Phi_G\geq \phi$.

For any vertex subset $S\subseteq V$ and positive value $x>0$, let $G[S]^x$ be the subgraph induced on $S$ where we add $\lceil x\rceil$ self-loops to each vertex $v\in S$ for every boundary edge $(v, w), w\notin S$. As an example, in graph $G[S]^1$ the degrees of all vertices are the same as in the original graph $G$.

We need a strong expander decomposition algorithm from a recent work~\cite{goranci2021expander}.
\begin{definition}[boundary-linked expander decomposition \cite{goranci2021expander}]\label{boundary}
	Let $G = (V, E)$ be a graph on $n$ vertices and $m$ edges, and $\alpha, \phi\in (0, 1)$ be parameters. An $(\alpha, \phi)$-expander decomposition of $V$ consists of a partition $\mathcal{C} = \{C_1, C_2, \cdots, C_k \}$ of $V$ such that the following holds.
	\begin{enumerate}[(1)]
		\item $\sum_{i=1}^k\out_G(C_i)\leq \log^4n \cdot \phi m$.
		\item For any $i$, $G[U_i]^{\alpha / \phi}$ is a $\phi$-expander.
		\item For any $i$, $\out_G(C_i)\leq \log^7n \cdot\phi\vol_G(C_i)$.
	\end{enumerate}
\end{definition}
In the original paper~\cite{goranci2021expander}, the upper bounds are stated with $O(\cdot)$ notations that hide constant factors; here we simply raise the exponent of log-factors to simplify the notations.

\begin{lemma}[\cite{goranci2021expander}]
	Given any unweighted graph $G= (V, E)$ on $n$ vertices and $m$ edges, for any $\alpha, \phi\in (0, 1)$ such that $\alpha \leq 1 / \log^\const m$ where $\const$ is a certain constant, a $(\alpha, \phi)$-expander decomposition can be computed in $\widetilde{O}(m / \phi)$ with high probability.
\end{lemma}

\subsection{Isolating cuts}
\begin{lemma}[isolating cuts~\cite{abboud2020subcubic}]\label{isolate}
	Given an undirected edge-weighted graph $H = (X, F, \omega)$, a pivot vertex $p\in X$, and a set of terminal vertices $T\subseteq X$. For each $u\in T$, let $(K_u, X\setminus K_u)$ be the latest minimum $(u, p)$-cut for each $u\in T$. Then, in time $\widetilde{O}(\mf(|F|, |X|))$ we can compute $|T|$ disjoint sets $\{K^\prime_u\}_{u\in T}$ such that for each $u\in T$, if $K_u\cap T = \{u\}$ then $K^\prime_u = K_u$.
\end{lemma}

%% file: conditional.tex
\section{Quadratic time cut-equivalent  tree under Hypothesis~\ref{hypo}}
\subsection{The main algorithm}
In this section we try to prove the first half of Theorem~\ref{cond}. Let $G = (V, E)$ denote the simple graph as our input data. Define some parameters: $\phi = \frac{1}{10\log^{\const + 10}n}$ is a global conductance parameter that is used to construct expander decompositions, and $r = 10\log^5 n$ is a sampling parameter which is needed when choosing pivots; here $\const$ is the same constant as in Definition~\ref{boundary}. Without loss of generality, assume $\sqrt{n}$ is an integral power of $2$. Define a degree set $\mathcal{D} = \{\sqrt{n},  2\sqrt{n}, 2^2\sqrt{n}, \cdots, n \}$.

\paragraph*{Preparation} Throughout the algorithm, $\gh$ will be the cut-equivalent tree under construction, where each of $\gh$'s node will represent a subset of vertices of $V$. As a preparation step, compute a $(\phi, \phi)$-expander decomposition on $G$ and obtain a partitioning $\mathcal{C} = \{C_1, C_2, \cdots, C_k \}$ of $V$. Categorize clusters in $\mathcal{C}$ according to their sizes: for each $2^i$, define $\mathcal{C}_{i}$ to be the set of clusters whose sizes are within interval $[2^i, 2^{i+1})$.

At the beginning, initialize $\gh$ to be a $\sqrt{n}$-partial tree by applying the algorithm from~\cite{hariharan2007mn} that takes running time $\widetilde{O}(n^2)$. 

\paragraph*{Iteration} In each round, we will divide simultaneously all nodes of $\gh$ which contains at least $20r$ vertices in $V$. In the end, the total number of rounds will be bounded by $\widetilde{O}(1)$. To describe our algorithm, let us focus on any single node $N\subseteq V$ of $\gh$ whose size $|N|$ is at least $20r$. The first step is to refine the partition of $N$ by a set of random pivots. More specifically, sample a pivot subset $R\subseteq N$ of size $10r$ by picking each vertex with probability proportional to its degree in $G$; more precisely, repeatedly sample for $10r$ times a vertex from $N$ where each vertex $v\in N$ is selected with probability $\deg_G(v) / \vol_G(N)$.

Then, refine the node $N$ of $\gh$ by computing a partial tree with respect to $R$ using Lemma~\ref{refine}, which further divides $N$ into several subsets each containing a distinct vertex from $R$. After applying this pivot-sampling \& refining step to each of the original node of $\gh$, $\gh$ has undergone one pass of partition, and now each node $U$ of $\gh$ is associated with a unique pivot vertex $p\in U\cap R$.

For the rest, let us focus on each node $U$ of the current $\gh$ as well as its pivot $p$, such that $U\subseteq N$ is a subdivision of the previous node $N$ but $\vol_G(U)\geq 0.5\vol_G(N)$. For each $k\in \mathcal{D}$, define $U_k = \{u\in U\mid \deg_G(u)\in [k, 2k) \}$, so $U = \bigcup_{k\in \mathcal{D}}U_k$. Take a parameter $d\in \mathcal{D}$ such that $d|U_d|$ is maximized, namely $d\in\arg\max_{k\in \mathcal{D}} k|U_k|$. Therefore, $2d|U_d|\geq \vol_G(U)/\log n$. Next, for each index $i$, define $\cnt[i]$ to be the number of vertices from $U_d$ that lie within clusters from $\mathcal{C}_i$. Take $s = 2^{i_U}$ such that $\cnt[i_U]$ is maximized.

For each cluster $C\in \mathcal{C}_{i_U}$ such that $C\cap U_d$ is nonempty, conduct the expander search routine described in the next subsection (Algorithm~\ref{exp-search}) to compute cuts $(K_u, V_\gh[U]\setminus K_u)$ in the auxiliary graph $G_\gh[U]$ for a set $W_C$ of vertices $u\in W_C\subseteq C\cap U_d$ with respect to pivot $p$; so $u\in K_u, p\in V_\gh[U]\setminus K_u$. It will be guaranteed that for each $u\in W_C$, we have $\vol_G(K_u\cap U)\leq 0.5\vol_G(N)$. In the end, define $W = \bigcup_{C\in\mathcal{C}_{i_U}}W_C$.

We will prove that all $(K_u, V_\gh[U]\setminus K_u)$ are latest min-cuts for $(u, p)$. Since all latest minimum cuts $(K_u, V_\gh[U]\setminus K_u)$ are with respect to $p$, they should form a laminar family. Then, for each $u\in W$ such that $K_u$ is maximal in the laminar family, split $K_u\cap U$ off the node $U$ and create a new node for vertex set $K_u\cap U$. Since we always take maximal $K_u$'s, all of these sets are disjoint in $V_\gh[U]$, so the creation of new nodes on $\gh$ is well-defined. Pseudo-code \textsf{CondGomoryHu} summarizes our algorithm.

\begin{algorithm}
	\caption{\textsf{CondGomoryHu}$(G = (V, E))$}
	initialize a partition tree $\gh$, as well as parameters $\phi, r$\;
	\While{$\exists N\subseteq V$, $U$ a node of $\gh$, $|N| \geq 20r$}{
		\For{node $N$ of $\gh$ with $|N|\geq 20r$}{
			repeat for $10r$ times: each time we sample a vertex $u\in N$ with probability $\frac{\deg_G(u)}{\vol_G(N)}$, and let the sampled set be $R$\;
			call Lemma~\ref{refine} on node $N$ with respect to $R$\;
			\For{node $U\subseteq N$ of $\gh$ such that $\vol_G(U) > 0.5\vol_G(N)$}{
				take $d\in\mathcal{D}$ such that $d|U_d|$ is maximized\;
				take $s = 2^{i_U}$ such that $\cnt[i_U]$ is maximized\;
				\For{each $C\in \mathcal{C}_{i_U}$}{
					run expander search on $C$ within node $U$ to compute a subset $W_C\subseteq C\cap U_d$, and the latest min-cuts $(K_u, V_\gh[U]\setminus K_u)$ for each $u\in W_C$\;
				}
				define $W = \bigcup_{C\in\mathcal{C}_{i_U}}W_C$\;
				for each $u\in W$ such that $K_u$ is maximal, split $K_u\cap U$ off of $U$ and create a new node on $\gh$\;
			}
		}
	}
	\For{node $N$ of $\gh$ such that $|N| < 20r$}{
		repeatedly refine $N$ using the generic Gomory-Hu steps until all nodes are singletons\;
	}
	\Return $\gh$ as a cut-equivalent tree\;
\end{algorithm}

%% file: expander.tex
\subsection{Finding latest min-cuts in expanders}
Our algorithm is similar to the one from \cite{abboud2020subcubic}. The input to this procedure is a node $U\subseteq V$ of the current partition tree $\gh$ under construction, together with parameters $s, d$ defined previously, as well as an expander $C\in \mathcal{C}_{i_U}$ that intersects $U_d$. The output of this procedure will be a subset $W_C\subseteq  C\cap U_d$, and their cuts $(K_u, V_\gh[U]\setminus K_u)$ for all $u\in W_C$ in the auxiliary graph $G_\gh[U]$, with the extra property that $\vol_G(K_u\cap U)\leq 0.5\vol_G(N)$. In the end, we will show that with high probability, all these cuts are latest min-cuts in $G_\gh[U]$.

To describe our algorithm, consider all vertices $u\in C\cap U_d$ and their latest minimum cuts $(L_u, V_\gh[U]\setminus L_u)$ with respect to pivot $p\in U$. Let $\lambda_u$ be the cut value of $(L_u, V_\gh[U]\setminus L_u)$ in $G_\gh[U]$. All of the sets $L_u\cap C\cap U_d$ should form a laminar family, which corresponds to a tree structure $\gh_p^d[C]$, where each tree node $M$ of $\gh_p^d[C]$ packs a subset of $C\cap U_d$, such that for all $u\in M$ the set $L_u\cap C\cap U_d$ are the same. More specifically, $\gh_p^d[C]$ is constructed as follows: first arrange the laminar family $\{L_u\cap C\cap U_d \}_{u\in C\cap U_d}$ as a tree, and then for each node $L_u\cap C\cap U_d$ on this tree, associate with this node the set $M = \{v\in C\cap U_d\mid L_v\cap C\cap U_d = L_u\cap C\cap U_d \}\subseteq L_u\cap C\cap U_d$.

We need to emphasize that our algorithm does not know $\gh_p^d[C]$ at the beginning, but it will gradually explore part of $\gh_p^d[C]$ during the process.

\paragraph*{Preparation} Initialize $W_C = \emptyset$. Assume $|C\cap U_d| \geq 10/\phi^2$; otherwise we could simply reset $W_C = C\cap U_d$ and run $|C\cap U_d|$ instances of max-flow to compute all latest cuts. As a preparatory step, the algorithm repeatedly takes a random subset $T\subseteq C\cap U_d$ by selecting each vertex independently with probability $\phi$. Then apply Lemma~\ref{isolate} on graph $G_\gh[U]$ to compute isolating cuts of terminal vertices from $T$ with respect to pivot $p$. This procedure goes on for $10\log n /\phi$ iterations, and for each $u\in C\cap U_d$, let $(A_u^i, V_\gh[U]\setminus A_u^i)$ be the isolating cut computed for $u$ in the $i$-th iteration; if $u$ was not selected by $T$ in the $i$-th iteration, simply set $A_u^i = \{u\}$. Finally, let $A_u$ be the set among $\{A_u^i \}_{1\leq i\leq 10\log n / \phi}$ such that the cut value of $(A_u, V_\gh[U]\setminus A_u)$ in the auxiliary graph $G_\gh[U]$ is minimized; to break ties, we select $A_u^i$ that minimizes $|A_u^i|$. Let $\kappa_u$ be the cut value of $(A_u, V_\gh[U]\setminus A_u)$ for $u\in C\cap U_d$.

If one of $|A_u\cap C\cap U_d| > 2/\phi$, then the algorithm fails and aborts; we will prove that the failure probability is small.

\paragraph*{Exploring $\gh_p^d[C]$} A node $M$ of $\gh_p^d[C]$ is called \textbf{large} if $|C\cap U_d\setminus L_u| \leq 2/\phi$ for any $u\in M$, and if $|L_u\cap C\cap U_d|\leq 2/\phi$ it is called \textbf{small}. Similarly, a vertex $u\in C\cap U_d$ is called large, if $|C\cap U_d\setminus L_u| \leq 2 / \phi$; otherwise if $|L_u\cap C\cap U_d|\leq 2 / \phi$, it is called small.

\begin{lemma}
All large nodes on $\gh_p^d[C]$ should lie on a single path ended at root.
\end{lemma}
\begin{proof}
	Suppose otherwise there exists two different large nodes $M_1, M_2$ of $\gh_p^d[C]$ such that $M_1\cap M_2 = \emptyset$. Take any $u_1\in M_1, u_2\in M_2$. Since $M_1, M_2$ are large nodes, we know $|C\cap U_d\setminus L_{u_1}|\leq 2 / \phi$, $|C\cap U_d\setminus L_{u_2}|\leq 2 / \phi$. As $C\cap U_d\cap L_{u_1}$ and $C\cap U_d\cap L_{u_2}$ are disjoint, we have $|C\cap U_d| \leq |C\cap U_d\setminus L_{u_1}| + |C\cap U_d\setminus L_{u_2}| \leq 4 / \phi$, which contradicts that $|C\cap U_d|\geq 10 / \phi^2$.
\end{proof}

The main idea of our algorithm is to find the lowest large node on $\gh_p^d[C]$; to clarify a bit more, here ``lowest'' means farthest from root. Initialize a set $S \leftarrow C\cap U_d$ and maintain an ordering of vertices in $S$ according to the cut value of $\kappa_u$, also initialize variable $Q \leftarrow \emptyset$. 

Repeat the following procedure: take $u\in S\setminus Q$ such that $\kappa_u$ is maximized. Apply max-flow in graph $G_\gh[U]$ to compute the latest min-cut $L_u$ for $(u, p)$. Consider two possibilities.
\begin{itemize}
	\item $L_u$ is small. Then assign $S\leftarrow S\setminus (Q\cup L_u)$, and $Q\leftarrow \emptyset$.
	\item $L_u$ is large. If $L_u\cap C\cap U_d = S$, then add $u$ to $Q$; otherwise if $L_u\cap C\cap U_d\neq S$, reset $S \leftarrow L_u\cap C\cap U_d$ and $Q\leftarrow \{u\}$.
\end{itemize}

The repetition terminates if either (1) $|C\cap U_d \setminus S| > 2/\phi$ or (2) $|Q| > 2/\phi$. In the first case, assign $W_C\leftarrow \{u\in S\mid \vol_G(A_u\cap U)\leq 0.5\vol_G(N) \}, K_u\leftarrow A_u$ and terminate. Note that this notation $\vol_G(A_u\cap U)$ is well-defined, since all vertices in $U$ are also vertices in $V$, not contracted vertices in $V_\gh[U]$.

Now suppose we are in the second case. Let $L = L_v$ for an arbitrary $v\in Q$. We will prove afterwards that $(L, V_\gh[U]\setminus L)$ is the latest min-cut corresponding to the lowest large node. Let $\kappa$ be the cut value of $(L, V_\gh[U]\setminus L)$, and define $B = \{u\in S\mid \kappa_u > \kappa \}$. Assign $W_C\leftarrow \{u\in S\setminus B\mid \vol_G(A_u\cap U)\leq 0.5\vol_G(N) \}, K_u\leftarrow A_u$ and for each $u\in W_C$. 

After that, take an arbitrary $v\in Q$. If it satisfies that $\vol_G(L\cap U) \leq 0.5\vol_G(N)$, then update $W_C = W_C\cup \{v \}$ and $K_v\leftarrow L$. The whole exploration procedure is summarized as pseudo-code \textsf{ExploreTree}.

\begin{algorithm}
	\caption{\textsf{ExploreTree}$(U, p, d, C)$}
	\label{exp-search}
	prepare $A_u$ and $\kappa_u$ for all $u\in C\cap U_d$\;
	initialize $S\leftarrow C\cap U_d, Q\leftarrow \emptyset$\;
	\While{$\max\{|C\cap U_d\setminus S|, |Q| \}\leq 2/\phi$}{
		take $u\in \arg\max_{v\in S\setminus Q} \{\kappa_v\}$\;
		apply max-flow to compute $L_u$\;
		\If{$u$ is small}{
			$S\leftarrow S\setminus (Q\cup L_u)$, and $Q\leftarrow \emptyset$\;
		}\Else{
			\If{$L_u\cap C\cap U_d = S$}{
				$Q\leftarrow Q\cup \{u\}$\;
			}\Else{
				$S\leftarrow L_u\cap C\cap U_d, Q\leftarrow \{u\}$\;
			}
		}
	}
	\If{$|C\cap U_d\setminus S| > 2/\phi$}{
		\Return $W_C\leftarrow \{u\in S\mid \vol_G(A_u\cap U)\leq 0.5\vol_G(N) \}$, $K_u\leftarrow A_u, \forall u\in W_C$\;
	}\Else{
		define $B = \{u\in S\mid \kappa_u > \kappa \}$ where $\kappa$ is the cut value of $(S, V_\gh[U]\setminus S)$\;
		$W_C\leftarrow \{u\in S\setminus B\mid \vol_G(A_u\cap U)\leq 0.5\vol_G(N) \}$, $K_u\leftarrow A_u, \forall u\in W_C$\;
		draw an arbitrary vertex $v\in Q$ and set $L = L_v$\;
		\If{$\vol_G(L\cap U)\leq 0.5\vol_G(N)$}{
			assign $W_C \leftarrow W_C\cup\{v\}$ and $K_v\leftarrow L$\;
		}
		\Return $W_C$\;
	}
\end{algorithm}

%% file: correctness.tex
\subsection{Proof of correctness}
First we prove a basic property of isolating cuts, which is also used in \cite{abboud2020subcubic}.
\begin{lemma}[\cite{abboud2020subcubic}]\label{imbalance}
	For each $u\in C\cap U_d$, either $|L_u\cap C\cap U_d|\leq 2/\phi$ or $|C\cap U_d\setminus L_u|\leq 2/\phi$; namely each vertex is either large or small. Furthermore, with high probability, when $u$ is small, $A_u = L_u$.
\end{lemma}
\begin{proof}
	Since $\deg_G(u) < 2d$, the cut value of $(L_u, V_\gh[U]\setminus L_u)$ is smaller than $2d$. Unpack all contracted vertices of $V_\gh[U]$, and let $L^\prime_u\subseteq V$ be the set of vertices belonging to $L_u$ or contracted in $L_u$. Therefore, since $G_\gh[U]$ is a contracted graph of $G$, the cut value of $(L_u, V_\gh[U]\setminus L_u)$ is equal to the cut value of $(L_u^\prime, V\setminus L_u^\prime)$.
	
	Suppose otherwise that $|L_u^\prime \cap C\cap U_d| > 2/\phi$ and $|C\cap U_d\setminus L_u^\prime| > 2/\phi$. Then, by property (2) of the $(\phi, \phi)$-expander decomposition, $G[C]^1$ is a $\phi$-expander, and so the number of edges between $L_u^\prime\cap C$ and $C\setminus L_u^\prime$ is at least $\phi \cdot \min\{\vol_G(L_u^\prime\cap C), \vol_G(C\setminus L_u^\prime) \} > 2d$, which is a contradiction.
	
	Let us turn to the second half of the statement. Suppose $u$ is small, and so $|L_u\cap C\cap U_d|\leq 2 / \phi$. Then, since $T$ selects each vertex in $C\cap U_d$ with probability $\phi$, with probability $\phi\cdot (1 - \phi)^l > \phi/8$, $T\cap L_u\cap C = \{u\}$ is a singleton. In this case, by Lemma~\ref{isolate}, $A_u^i = L_u$. As $T$ is sampled for $10\log n / \phi$ times, with high probability $A_u = L_u$.
\end{proof}

Here is a basic fact regarding large vertices.
\begin{lemma}\label{basic}
	For any large vertex $u$, $\lambda_u < \kappa_u$.
\end{lemma}
\begin{proof}
	If $\lambda_u = \kappa_u$, then the latest min-cut should be contained in $A_u$, which contains at most $2/\phi$ vertices from $C\cap U_d$, and so $u$ cannot be large.
\end{proof}

Next we analyze the behavior of the while-loop in \textsf{ExploreTree}.
\begin{lemma}\label{Mset}
	If $Q\neq\emptyset$, then for each $u\in Q$, $L_u\cap C\cap U_d = S$.
\end{lemma}
\begin{proof}
	Each time $S$ is updated, either $Q$ adds a vertex $u$ on line-10 such that $L_u\cap C\cap U_d = S$, or $Q$ is updated to $\{u\}$ on line-12. So the equality always holds.
\end{proof}

\begin{lemma}\label{explore}
	At the beginning of any iteration of the while-loop, $\forall v\in S$, if $v$ is large, then we have $L_v\cap C\cap U_d \subseteq S$.
\end{lemma}
\begin{proof}
	We prove this statement by induction on the number of iterations. Initially, this holds as $S = C\cap U_d$. For any intermediate iteration, consider two cases.
	\begin{itemize}
		\item $u$ is small. We claim that before updating $S$, for all large vertices $v\in S\setminus (Q\cup L_u)$, $L_v$ and $L_u\cup Q$ are disjoint; if this can be proved, then we conclude $L_v\cap C\cap U_d\subseteq S\setminus (Q\cup L_u)$, as $L_v\cap C\cap U_d\subseteq S$ holds before.
		
		Suppose that $L_v\cap L_u\neq \emptyset$. Then as all latest minimum cuts form a laminar family and that $v\notin L_u$, it must be $L_u\subseteq L_v$. As $v$ is large, $L_v\cap C\cap U_d$ contains more vertices than $A_v\cap C\cap U_d$, and so by Lemma~\ref{basic}, we have $\lambda_v < \kappa_v$. Now, by line-4, since $\kappa_u$ is the largest among all vertices in $S\setminus Q$, $\kappa_v\leq \kappa_u$. Finally, using Lemma~\ref{imbalance}, we know $\kappa_u = \lambda_u$ as $u$ is small. Concatenating all the inequalities we have: \[\lambda_v < \kappa_v\leq \kappa_u  = \lambda_u\]
		which contradicts the fact that $(L_u, V_\gh[U]\setminus L_u)$ is a min-cut for $(u, p)$ as $L_u\subseteq L_v$.
		
		Now suppose that $L_v\cap Q\neq \emptyset$, say $w\in L_v\cap Q$. Then by Lemma~\ref{Mset}, $v\in S\subseteq L_w$, and so both $w, v$ are in $L_v\cap L_w$, which means $L_v = L_w$, and so $L_v\cap L_u = L_w\cap L_u = L_u\neq \emptyset$, which is a contradiction as discussed just before.
		
		\item $u$ is large. In this case, the algorithm would reassign $S\leftarrow L_u\cap C\cap U_d$. Then, for all $v\in S$, as $(L_v, V_\gh[U]\setminus L_v)$ is the latest minimum cut, it must be $L_v\subseteq L_u$, irrespective of whether $v$ is large or not.\qedhere
	\end{itemize}
\end{proof}

Next we prove that when the while-loop ends, either all vertices in $S$ are small, or $S$ corresponds to the cut of the lowest large node on $\gh_p^d[C]$.

\begin{lemma}\label{all-small}
	If $|C\cap U_d\setminus S| > 2/\phi$, then all vertices in $S$ are small.
\end{lemma}
\begin{proof}
	Consider any vertex $u\in S$. If $u$ is large, then by Lemma~\ref{explore}, $L_u\cap C\cap U_d\subseteq S$, and so $|C\cap U_d\setminus L_u|\geq |C\cap U_d\setminus S| > 2/\phi$, which contradicts the definition of being large.
\end{proof}

\begin{lemma}\label{lowest}
	After the while-loop ends, if $|C\cap U_d\setminus S| \leq 2/\phi$ and $|Q| > 2/\phi$, then $(L, V_\gh[U]\setminus L)$ is the latest min-cut of the lowest large node on $\gh_p^d[C]$. Moreover, $B = \{u\in S\mid \kappa_u >\kappa \}$ is the set of all vertices $u$ such that $L_u\cap C\cap U_d = S$, and consequently all $L_u, \forall u\in B$ are equal.
\end{lemma}
\begin{proof}
	As the while-loop ends with $|Q| > 2/\phi$, the last iteration must have ended on line-10. Therefore, $(L, V_\gh[U]\setminus L)$ is the latest min-cut of some $u\in Q$. Suppose otherwise $(L, V_\gh[U]\setminus L)$ is not the latest min-cut of the lowest large node on the imaginary tree $\gh_p^d[C]$. Then, there exists a large vertex $v\in L\cap C\cap U_d$ such that $L_v\cap C\cap U_d\subsetneq S$ but $|C\cap U_d\setminus L_v| \leq 2/\phi$. As $|Q| > 2/\phi$, there must exist $w\in L_v\cap Q$. By Lemma~\ref{Mset}, $v\in L\cap C\cap U_d =S = L_w\cap C\cap U_d$, so both $v, w$ are in $L_v\cap L_w$, and consequently $L_v = L_w$, $L_v\cap C\cap U_d = S$, contradiction.
	
	Now let us turn to the second half of the statement. Consider any $u\in B$. $(A_u, V_\gh[U]\setminus A_u)$ cannot be a min-cut as $\kappa_u > \kappa$. By Lemma~\ref{imbalance}, $u$ must be a large vertex. On the one hand, by Lemma~\ref{explore}, $L_u\cap C\cap U_d\subseteq S$, and on the other hand, $L_u\cap C\cap U_d$ cannot be strictly smaller than $S$ as $S$ is the lowest already. Hence $L_u\cap C\cap U_d = S$.
	
	For any $u\notin B$, by definition $\kappa_u\leq \kappa$. If $u$ is large, then $\lambda_u<\kappa_u \leq \kappa$, so $L_u\cap C\cap U_d\subsetneq S$, which also contradicts that $S$ corresponds to the lowest large node on $\gh_p^d[C]$.
\end{proof}

Finally, we prove that all cuts $(K_u, V_\gh[U]\setminus K_u)$ output by the algorithm are latest min-cuts with high probability.
\begin{lemma}
	All cuts $(K_u, V_\gh[U]\setminus K_u)$ output by the algorithm are latest min-cuts with high probability.
\end{lemma}
\begin{proof}
	If the algorithm terminates on line-14, then by Lemma~\ref{all-small}, all vertices in $W_C$ are small. So by Lemma~\ref{imbalance}, $L_u = A_u = K_u, \forall u\in W_C$. Otherwise, if the algorithm terminates on line-21, then by Lemma~\ref{lowest}, all vertices in $W_C$ are small. Hence, by Lemma~\ref{imbalance}, $L_u = A_u = K_u, \forall u\in W_C\setminus B$; also, for any $u\in W_C\cap B$, we have $L_u = L = K_u$.
\end{proof}

%% file: time.tex
\subsection{Running time analysis}
First we analyze the running time of each call of expander search.
\begin{lemma}\label{expander-search}
	The total running time of the expander search in graph $G_\gh[U]$ is bounded by $$\widetilde{O}(\mf(\vol_G(U), |V_\gh[U]|) / \phi)$$
\end{lemma}
\begin{proof}
	During the preparation step, each invocation of Lemma~\ref{isolate} induces a set of max-flow instances whose total size is bounded by $\widetilde{O}(|E_\gh[U]|) = \widetilde{O}(\vol_G(U))$. Since it is repeated for $O(\log n / \phi)$ times, the total time is at most $\widetilde{O}(\mf(\vol_G(U), |V_\gh[U]|) / \phi)$.
	
	Next, let us analyze the cost of \textsf{ExploreTree}.
	\begin{claim}
		After each iteration of the while-loop, the value of $|C\cap U_d\setminus S| + |Q|$ always increases by at least $1$, so the total number of max-flow instances during the loop is bounded by $O(1/\phi)$.
	\end{claim}
	\begin{proof}[Proof of claim]
		If an iteration ends on line-10, $Q$ increases by one while $S$ does not change. If an iteration of the while-loop ends on line-7, then on the one hand, by Lemma~\ref{Mset} we have $Q\subseteq S$; on the other hand, by the pseudo-code, $u\notin Q$ before updating $S, Q$. Hence, after line-7, $|C\cap U_d\setminus S| + |Q|$ increases by at least $1$.
	
		If an iteration ends on line-12, we claim that before updating $S, Q$, we have $L_u \cap Q = \emptyset$. In fact, by Lemma~\ref{Mset}, for any $w\in Q$, $L_w\cap C\cap U_d = S$. By Lemma~\ref{explore}, as $L_u\cap C\cap U_d\neq S$, it must be $L_u\cap C\cap U_d\subsetneq S = L_w\cap C\cap U_d$. Hence, $w\notin L_u$. As $w$ is arbitrary, we know $Q\cap L_u = \emptyset$. Therefore, after updating $S\leftarrow L_u\cap C\cap U_d$, $|C\cap U_d\setminus S|$ has increased by $|Q|$. Notice that after updating $Q$, $|Q| = 1$. So $|C\cap U_d\setminus S| + |Q|$ has increased by one.
	\end{proof}
	
	Since each while-loop conducts one max-flow in graph $G_\gh[U]$, by the above claim, the total cost of the while-loop involves max-flow instances of total size  $\widetilde{O}(\vol_G(U) / \phi)$, and the reduction time is dominated by the same amount. After the while-loop, the running time is linear in the size of output, so it is not the bottleneck.
\end{proof}

Next we analyze the running time during refinement of $U$.
\begin{lemma}\label{prune}
	The total running time of cutting vertices (the for-loop on line-6 of \textsf{CondGomoryHu}) from $U$ takes total time of $\widetilde{O}(\frac{n}{s}\cdot\mf( 2d\cdot\cnt[i_U]\log^2n, |V_\gh[U]|))$.
\end{lemma}
\begin{proof}
	On the one hand, the number of clusters in $\mathcal{C}_{i_U}$ is at most $n/s$ since each cluster has size at least $s$. So, by Lemma~\ref{expander-search}, the total time of expander search is $\widetilde{O}(\frac{n}{s}\cdot \mf(\vol_G(U)))$. By maximality of $d|U_d|$ and $\cnt[i_U]$, we have that: \[\vol_G(U)\leq 2d|U_d|\log n \leq 2d\cdot \cnt[i_U]\log^2n\]
	
	Since the number of edges in $G_\gh[U]$ is $\vol_G(U)$, the overall time complexity would be $\widetilde{O}(\frac{n}{s}\cdot\mf( 2d\cdot\cnt[i_U]\log^2n, |V_\gh[U]|))$.
\end{proof}

To bound the total time across all different nodes of $\gh$ that correspond to the same choice of $(s, d)$, we need the following lemma.
\begin{lemma}\label{count}
	In any single iteration of the while-loop on line-2 of \textsf{CondGomoryHu}, over all different nodes $U$ of $\gh$ that correspond to the same choice of $(s, d)$, we have $\sum_{U}\cnt[i_U]\leq 4sn/d$.
\end{lemma}
\begin{proof}
	If $s\geq d/4$, then since all such nodes $U$ are packing disjoint subsets of vertices of $V$, $\sum_{U}\cnt[i_U]\leq n\leq 4sn/d$. So next we only consider the case where $s < d/4$.
	
	When $s < d/4$, we can upper bound the total number of vertices in clusters in $\mathcal{C}_{i_U}$ whose degrees in $G$ are within the interval $[d, 2d)$. Take any cluster $C\in \mathcal{C}_{i_U}$ and any vertex $u\in C\cap U_d$. Since $|C| < 2s = d/2$ and $G$ is a simple graph, at least $d/2$ of $u$'s neighbors in $G$ are outside of $C$. So the crossing edges contributed by $u$ is at least $d/2$. By property (3) of the Definition~\ref{boundary}, the total number of crossing edges should be bounded as:
	\[\begin{aligned}
	\sum_{C\in\mathcal{C}_{i_U}}\out_G(C) &\leq \log^7n\cdot \phi\cdot \sum_{C\in\mathcal{C}_{i_U}}\vol_G(C) \leq \log^7n\cdot \phi\cdot \sum_{C\in\mathcal{C}_{i_U}} (4s^2 + \out_G(C))\\
	&\leq 4\log^7n\cdot \phi sn + \log^7n\cdot \phi  \sum_{C\in\mathcal{C}_{i_U}}\out_G(C)
	\end{aligned}\]
	As $\phi = \frac{1}{10\log^{\const+10}n}$, we have $\sum_{C\in\mathcal{C}_{i_U}}\out_G(C) \leq 8\log^7n\cdot \phi sn < sn$. As each vertex in $C\cap U_d$ contributes $d/2$ to the above summation, the total number of vertices from $U_d$ in $\mathcal{C}_{i_U}$ is bounded by $2sn/d$.
\end{proof}

Combining the above two lemmas gives the following corollary.
\begin{corollary}
	Under Hypothesis~\ref{hypo}, the time of dividing all nodes of $\gh$ for a single iteration of the while-loop on line-2 in \textsf{CondGomoryHu} is bounded by $\widetilde{O}(n^2)$.
\end{corollary}

The next thing would be analyzing the total number of rounds of the while-loop. Similar to~\cite{abboud2020subcubic}, we first need to prove that with high probability, the number of $u\in C\cup U_d$ such that $\vol_G(K_u) > 0.5\vol_G(U)$ is roughly at most $|U_d| / r$.
\begin{lemma}\label{pivot}
	With high probability over the choice of $R\subseteq N$, the total number of $u\in U_d$ such that $\vol_G(L_u\cap U) > 0.5\vol_G(N)$ is at most $\frac{4|U_d|\log^2n}{r}$.
\end{lemma}
\begin{proof}
	The proof is similar to the one in~\cite{abboud2020subcubic}. To avoid confusion, let $\gh^\old$ be the version of $\gh$ before refining with respect to $R$, and let $\gh$ refer to the tree after refinement. For each pair of vertices in the super node $x, q\in N$, define $(\Gamma_x^q, V_{\gh^\old}[N]\setminus \Gamma_x^q)$ to be the latest minimum cut of $(x, q)$ in $G_{\gh^\old}[N]$. Define $M_x^q$ to be the set of all vertices $y\in N$ such that $x\in \Gamma_y^q$. Basic concentration inequalities show that for any $q, x$, if $\vol_G(M_x^q) \geq \frac{\log n}{r}\vol_G(N)$, then with high probability, $M_x^q\cap R\neq \emptyset$.
	
	The following claim is a crucial relationship between $\Gamma_u^p$ and $L_u$.
	\begin{claim}[Observation 4.3 in \cite{abboud2020subcubic}]
		$\Gamma_u^p\cap U = L_u\cap U$.
	\end{claim}
	\begin{proof}[Proof of claim]
		As $(L_u, V_\gh[U]\setminus L_u)$ is the latest minimum cut in $G_\gh[U]$ which is a contracted graph of $G_{\gh^\old}[N]$ following standard Gomory-Hu steps, $(L_u, V_\gh[U]\setminus L_u)$ is a min-cut for $(u, p)$ in $G_{\gh^\old}[N]$. Since $(\Gamma_u^p, V_{\gh^\old}[N]\setminus \Gamma_u^p)$ is the latest minimum cut in $G_{\gh^\old}[N]$, we have $\Gamma_u^p\cap U\subseteq L_u\cap U$. Next we only focus on the other direction.
		
		Let $W_1, W_2, \cdots, W_l\subseteq V_{\gh^\old}[N]$ be all contracted vertices of $V_\gh[U]$ which are crossed by $\Gamma_u^p$; in other words, $\Gamma_u^p\cap W_i\neq \emptyset$ and $W_i\setminus \Gamma_u^p\neq \emptyset$ for all $1\leq i\leq l$. Since $N$ is refined using pivots from $R$, according to Lemma~\ref{partial-tree}, we know that for each $i$ there exists a pivot $q_i\in R\cap W_i$ such that $(W_i, V_{\gh^\old}[N]\setminus W_i)$ is a minimum cut for $(q_i, p)$; in fact, $q_i\in R\cap W_i$ are in the neighboring nodes of $U$ in $\gh$.
		
		We claim that $q_i\in \Gamma_u^p$; otherwise if $q_i\notin \Gamma_u^p$, as $u\notin W_i$, by Lemma~\ref{minus}, the cut $(X, V_{\gh^\old}[N]\setminus X)$ where $X = \Gamma_u^p\setminus W_i$ is also a minimum cut for $(u, p)$, which contradicts that $(\Gamma_u^p, V_{\gh^\old}[N]\setminus \Gamma_u^p)$ is the latest min-cut.
		
		Construct a new cut $(Y, V_{\gh^\old}[N]\setminus Y)$ where $Y = \Gamma_u^p \cup\bigcup_{i=1}^l W_i$. On the one hand, as $q_i\in \Gamma_u^p, \forall i$, by repeatedly applying Lemma~\ref{union} we know $(Y, V_{\gh^\old}[N]\setminus Y)$ is a minimum cut for $(u, p)$ as well; On the other hand, $Y$ does not cross any contracted nodes in $G_\gh[U]$, so $(Y, V_\gh[U]\setminus Y)$ is a valid cut in $G_\gh[U]$ as well. As $(L_u, V_\gh[U]\setminus L_u)$ is the latest cut in $G_\gh[U]$, we know $L_u\cap U\subseteq Y\cap U = \Gamma_u^p\cap U$. This concludes our proof.
	\end{proof}
	
	Consider the set of all $u\in U_d$ such that $\vol_G(L_u\cap U) > 0.5\vol_G(N)$; let them be $u_1, u_2, \cdots, u_l$. By the above claim, it must be $\vol_G(\Gamma_{u_i}^p\cap N)\geq \vol_G(\Gamma_{u_i}^p\cap U) = \vol_G(L_{u_i}\cap U) > 0.5\vol_G(N)$ as well. Therefore, any two sets $\Gamma_{u_i}^p, \Gamma_{u_j}^p$ must intersect. Since $(\Gamma_{u_i}^p, V_{\gh^\old}[N]\setminus \Gamma_{u_i}^p)$ are latest cuts with respect to the same pivot $p$, they should form a total order, say $\Gamma_{u_1}^p\subseteq \Gamma_{u_2}^p\subseteq\cdots \Gamma_{u_l}^p$, and so by definition $u_2, u_3, \cdots, u_l\in M_{u_1}^p$. 
	\begin{claim}
		$M_{u_1}^p\cap R = \emptyset$.
	\end{claim}
	\begin{proof}[Proof of claim]
		If $\exists w\in M_{u_1}^p\cap R$, then by definition, $u\in \Gamma_w^p$. As $(\Gamma_w^p, V_{\gh^\old}[N]\setminus \Gamma_w^p)$ is the latest min-cut for $(w, p)$ in $G_{\gh^\old}[N]$, any min-cut for $(w, p)$ in $G_{\gh^\old}[N]$ should contain $u$ on the same side as $w$. By Lemma~\ref{partial-tree}, $u$ should belong to the part which contains $w$ after the refinement with respect to $R$, which makes a contradiction as $u$ stays with $p$ in the same part.
	\end{proof}
	By the above lemma, we know $\vol_G(M_{u_1}^p) < \frac{\log n}{r}\vol_G(N)$, and hence we have: 
	\[d(l-1) \leq \vol_G(M_{u_1}^p) <\frac{\log n}{r}\vol_G(N)\leq \frac{2\log n}{r}\vol_G(U)\leq \frac{4\log^2n}{r}d|U_d|\]
	So $l\leq \frac{4|U_d|\log^2n}{r}$.
\end{proof}

Finally we need to bound the total number of rounds in the while-loop. Call a cluster $C\in\mathcal{C}_{i_U}$ \textbf{bad}, if the total number of vertices $u\in C\cap U_d$ such that $\vol_G(L_u\cap U) > 0.5\vol_G(N)$ is more than $0.1|C\cap U_d|$; otherwise it is called \textbf{good}.
\begin{lemma}\label{cutoff}
	Consider any invocation of \textsf{ExploreTree} with input parameters $U, p, d, C$. Suppose cluster $C$ is good, then $|W_C| > 0.8|C\cap U_d|$.
\end{lemma}
\begin{proof}
	First consider the case where \textsf{ExploreTree} terminated on line-14. The while-loop must have terminated on line-7. Then as $u$ is small, $|S|\geq |C\cap U_d| - 2/\phi - |Q| - |C\cap L_u\cap U_d| \geq |C\cap U_d| - 6/\phi$. Therefore, $|W_C|\geq |C\cap U_d| - 6/\phi - 0.1|C\cap U_d| > 0.8|C\cap U_d|$.
	
	Now suppose \textsf{ExploreTree} terminated on line-21. In this case, $C\cap U_d \setminus W_C$ only includes vertices in $C\cap U_d\setminus S$, plus vertices $v\in C\cap U_d$ such that $\vol_G(L_v\cap U) > 0.5\vol_G(N)$. Since $C$ is good, we know $|W_C| \geq |C\cap U_d| - 2/\phi - 0.1|C\cap U_d| > 0.8|C\cap U_d|$.
\end{proof}

\begin{lemma}
	 $\sum_{C\in\mathcal{C}_{i_U}\text{ is bad}}|C\cap U_d|\leq \frac{40|U_d|\log^2n}{r}$.
\end{lemma}
\begin{proof}
	By Lemma~\ref{pivot}, the total number of vertices $u$ such that $\vol_G(L_u\cap U) > 0.5\vol_G(N)$ is at most $\frac{4|U_d|\log^2n}{r}$. By definition of badness, we have $\sum_{C\in\mathcal{C}_{i_U}\text{ is bad}}|C\cap U_d|\leq \frac{40|U_d|\log^2n}{r}$.
\end{proof}

\begin{lemma}\label{depth}
	For each node $U$ such that $\vol_G(U) > 0.5\vol_G(N)$, and for each set $K_u = L_u$ which is cut off by our algorithm, we have $\vol_G(K_u\cap U)\leq 0.5\vol_G(N)$. Furthermore, let $P$ be the rest of $U$ after cutting all $K_u$'s. Then $\vol_G(P)\leq (1 - \frac{1}{2\log^2n})\vol_G(U)$.
\end{lemma}
\begin{proof}
	The first half of the claim is automatically guaranteed by the algorithm. Let us only consider the second half.
	
	By Lemma~\ref{cutoff}, the total volume that has been cut off from $U$ is at least 
	\[\begin{aligned}
	\sum_{C\in\mathcal{C}_{i_U}\text{ is good}}d|W_C| &\geq \sum_{C\in C_{i_U}\text{ is good}}0.8d|C\cap U_d| \geq \frac{0.8d}{\log n}|U_d| - 0.8d\sum_{C\in\mathcal{C}_{i_U}\text{ is bad}}|C\cap U_d|\\
	&\geq \frac{0.8d}{\log n}|U_d| - \frac{32d\log^2n}{r}|U_d|\geq \frac{0.8}{\log^2n}\vol_G(U) - \frac{32}{\log^3n}\vol_G(U)\\
	&\geq \frac{1}{2\log^2n}\vol_G(U)
	\end{aligned}\]
	
	Hence, the volume of $\vol_G(P)$ is reduced by a factor of at most $1 - \frac{1}{2\log^2n}$.
\end{proof}

By Lemma~\ref{depth}, after each round of the while-loop, for each node $U\subseteq N$, either we already have $\vol_G(U)\leq 0.5\vol_G(N)$ after the refinement with respect to random set $R$, or $U$ is further divided into sub-nodes whose volume are at most $\max\{0.5\vol_G(N), (1 - \frac{1}{2\log^2n})\vol_G(U)\}$. Therefore the number of rounds is at most $\log^3n$. So the total running time should be $\widetilde{O}(n^2)$ as well under Hypothesis~\ref{hypo}.

%% file: unconditional.tex
\section{Unconditional cut-equivalent trees}
\subsection{The main algorithm}
In this section we will prove the second half of Theorem~\ref{cond} using existing max-flow algorithms. The algorithm is mostly the same as the previous algorithm conditioning on Hypothesis~\ref{hypo}, and the extra work is to deal with the additive $n^{1.5}$ term that appears in the running time of max-flow algorithm from~\cite{brand2021minimum}. Similar to the previous algorithm, we will also use the same set of parameters $\phi, r$, and use degree set $\mathcal{D} = \{\sqrt{n},  2\sqrt{n}, 2^2\sqrt{n}, \cdots, n \}$.

\paragraph*{Preparation} Throughout the algorithm, $\gh$ will be the cut-equivalent tree under construction, where each of $\gh$'s node will represent a subset of vertices of $V$. As a preparation step, compute a $(\phi, \phi)$-expander decomposition on $G$ and obtain a partitioning $\mathcal{C} = \{C_1, C_2, \cdots, C_k \}$ of $V$. Categorize clusters in $\mathcal{C}$ according to their sizes: for each $2^i$, define $\mathcal{C}_{i}$ to be the set of clusters whose sizes are within interval $[2^i, 2^{i+1})$.

\paragraph*{Iteration} In each round, the algorithm tries to simultaneously subdivide all nodes of $\gh$ which contains at least $20r$ vertices in $V$. Following the same procedure as in the previous algorithm, for each node $N$ of $\gh$, further refine $N$ into a set of smaller sub-nodes. Then, for each such sub-node $U$, define variables $d, U_d$ and $s = 2^{i_U}$ accordingly. If (1) $d \geq n^{3/4}$ or (2) $s > n^{3/4} / \sqrt{d}$, we would continue to the do same as in algorithm \textsf{CondGomoryHu} which invokes the expander search procedure.

The unconditional algorithm diverges from the conditional algorithm in Theorem~\ref{cond} from here if $d < n^{3/4}$ and $s\leq n^{3/4} / \sqrt{d}$. Intuitively, when $s$ is relatively small, the number of expanders whose sizes are roughly $s$ would be large, and so expander searches would be costly because of the additive term $n^{1.5}$ in the running time of computing max-flow. What we would do is to directly apply Lemma~\ref{k-partial} on graph $G_\gh[U]$ to isolate all vertices in $U_d$ on the tree $\gh$ once and for all. The pseudo-code is summarized as \textsf{GomoryHu}.

\begin{algorithm}
	\caption{\textsf{GomoryHu}$(G = (V, E))$}
	initialize a partition tree $\gh$, as well as parameters $\phi, r$\;
	\While{$\exists N\subseteq V$, $U$ a node of $\gh$, $|N| \geq 20r$}{
		\For{node $N$ of $\gh$ with $|N|\geq 20r$}{
			repeat for $10r$ times: each time we sample a vertex $u\in N$ with probability $\frac{\deg_G(u)}{\vol_G(N)}$, and let the sampled set be $R$\;
			call Lemma~\ref{refine} on node $N$ with respect to $R$\;
			
			\For{node $U\subseteq N$ of $\gh$ such that $\vol_G(U) > 0.5\vol_G(N)$}{
				take $d$ such that $d|U_d|$ is maximized\;
				take $s = 2^{i_U}$ such that $\cnt[i_U]$ is maximized\;
				\If{$d\geq n^{3/4}$ or $s > n^{3/4} / \sqrt{d}$}{
					\For{each $C\in \mathcal{C}_{i_U}$}{
						run expander search on $C$ within node $U$ to compute a subset $W_C\subseteq C\cap U_d$, and the latest min-cuts $(K_u, V_\gh[U]\setminus K_u)$ for each $u\in W_C$\;
					}
					define $W = \bigcup_{C\in\mathcal{C}_{i_U}}W_C$\;
					for each $u\in W$ such that $K_u$ is maximal, split $K_u\cap U$ off of $U$ and create a new node on $\gh$\;
				}\Else{
					apply Lemma~\ref{k-partial} on the auxiliary graph $G_\gh[U]$ with input parameter $k = 2d$, so that all vertices in $U_d$ become singletons in $\gh$\;
				}
			}
		}
	}
	\For{node $U$ of $\gh$ such that $|U| < 20r$}{
		repeatedly refine $U$ using the generic Gomory-Hu steps until all nodes are singletons\;
	}
	\Return $\gh$ as a cut-equivalent tree\;
\end{algorithm}

\subsection{Running time analysis}
\begin{lemma}
	Each round of the while-loop in \textsf{GomoryHu} takes time $\widetilde{O}(n^{17/8})$.
\end{lemma}
\begin{proof}
	Let us study an arbitrary iteration. Suppose the condition on line-6 holds, namely $d\geq n^{3/4}$ or $s>n^{3/4}/\sqrt{d}$. Then in this case we would do exactly the same as in the conditional algorithm, and the only difference we are invoking the max-flow algorithm from~\cite{brand2021minimum}. According to Lemma~\ref{prune}, we could upper bound the running time as \[\widetilde{O}(\frac{n}{s}\cdot\mf( 2d\cdot\cnt[i_U]\log^2n, |V_\gh[U]|)) = \widetilde{O}(\frac{nd}{s}\cnt[i_U] + \frac{n}{s}\cdot |V_\gh[U]|^{1.5})\]
	for each node $U$. Since all tree nodes $U$ are disjoint vertex subsets of $V$, $\sum_{U\in \gh}|V_\gh[U]|^{1.5} \leq n^{1.5}$. Therefore, by Lemma~\ref{count}, this sums to $\widetilde{O}(n^2 + \frac{n^{2.5}}{s})$.
	
	We first claim that $s \geq \sqrt{2d}$. In fact, by maximality of $\cnt[i_U]$, there exists at least one cluster $C\in \mathcal{C}_{i_U}$ that intersects $U_d$. Take any $u\in C\cap U_d$. Then since $G$ is a simple graph, more than $d - s$ neighbors of $u$ are outside of $C$, thus $\out_G(C) > d-s$. By property (3) of Definition~\ref{boundary}, we have:
	\[\out_G(C)\leq \log^7n\cdot \phi\vol_G(C) \leq \log^7n\cdot \phi (4s^2 + \out_G(C))\]
	As $\phi = \frac{1}{10\log^{\const + 10}}$, we have $\out_G(C)\leq 0.4s^2 + 0.1\out_G(C)$, and so $\out_G(C) < 0.5s^2$. As $\out_G(C) > d-s$, we have $s > \sqrt{2d}$.
	
	When $d\geq n^{3/4}$, as $s\geq \sqrt{2d} >n^{3/8}$ we have $\widetilde{O}(n^2 + \frac{n^{2.5}}{s}) = \widetilde{O}(n^{17/8})$. If $d < n^{3/4}$ and $s > n^{3/4} / \sqrt{d}$, then we also bound the total running time as $\widetilde{O}(n^2 + \frac{n}{s\phi}\cdot n^{1.5}) = \widetilde{O}(n^{17/8})$.
	
	Now suppose the condition on line-6 does not hold, then $d < n^{3/4}$ and $s\leq n^{3/4}/\sqrt{d}$. In this case, similar to Lemma~\ref{count}, we can prove that the total volume $\vol_G(U)$ over all different $U$'s is bounded by $\widetilde{O}(ns)$. So applying Lemma~\ref{k-partial} in this round takes time at most $\widetilde{O}(nsd) = \widetilde{O}(n^{17/8})$.
\end{proof}

\begin{lemma}
	The total number of rounds of the while-loop in \textsf{GomoryHu} is bounded by $O(\log^3n)$.
\end{lemma}
\begin{proof}
	If each round of the while-loop, if $d\geq n^{3/4}$ or $s > n^{3/4} / \sqrt{d}$ for node $U$, then according to the proof of Lemma~\ref{depth}, the volume of each subdivision is bounded by $\max\{0.5\vol_G(N), (1 - \frac{1}{2\log^2n})\vol_G(U)\}$. If $d< n^{3/4}$ and $s \leq n^{3/4} / \sqrt{d}$, then all vertices in $U_d$ become singletons on $\gh$; also, and for the same reason, all subdivision of $U$ should be at most $(1 - \frac{1}{2\log^2n})\vol_G(U)$. Therefore, the number of while-loop iterations is at most $O(\log^3n)$.
\end{proof}

%% file: ack.tex
\section*{Acknowledgment}
The author would like to thank helpful discussions with Prof. Ran Duan and Dr. Amir Abboud. This publication has received funding from the European Research Council (ERC) under the European Union’s Horizon 2020 research and innovation programme (grant agreement No 803118 UncertainENV).